\documentclass[11pt]{article}

\usepackage[utf8]{inputenc}

\usepackage{a4wide,amssymb,amsmath,amsthm,enumerate}
\usepackage{euler}
\usepackage{xspace}
\usepackage{tikz}
\usetikzlibrary{arrows,positioning}

\usepackage{ourmacros}

\title{EXPSPACE-hardness of behavioural equivalences
\\
of succinct one-counter nets}

\author{Petr Jan\v{c}ar$^1$ \and Petr Osi\v{c}ka$^1$ 
\and Zden\v{e}k Sawa$^2$
\\
 \\
$^1${\small Dept of Comp. Sci., Faculty of Science, Palack\'y Univ. Olomouc, Czech Rep.}\\
{\small pj.jancar@gmail.com, osicka@acm.org}
\\
$^2${\small Dept of Comp. Sci., FEI, Techn.~Univ.~Ostrava, Czech Rep.}\\
{\small zdenek.sawa@vsb.cz}
}

\date{}

\begin{document}

\maketitle

\begin{abstract}
We note that the remarkable EXPSPACE-hardness result
	in [G\"oller, Haase, Ouaknine, Worrell, ICALP 2010] ([GHOW10]
	for short)
	allows us
 to answer an open complexity question for simulation preorder 
of succinct one counter nets (i.e., one counter automata 
with no zero tests where counter increments and decrements
	are integers written
	in binary). 
This problem, as well as bisimulation equivalence, turn out to be 
	EXPSPACE-complete. 
	
The technique of [GHOW10]
was referred to
by Hunter [RP 2015] for deriving EXPSPACE-hardness 
of reachability games on succinct one-counter nets.
We first give a~direct self-contained EXPSPACE-hardness proof 
for such reachability games  (by adjusting a known PSPACE-hardness
proof for emptiness of alternating finite automata with one-letter
	alphabet); then we reduce reachability games to (bi)simulation
	games by using a~standard ``defender-choice'' technique.
\end{abstract}	

\section{Introduction}

We concentrate on our contribution, without giving 
 a broader overview of the area here.

A remarkable result by G\"oller, Haase, Ouaknine,
Worrell~\cite{DBLP:conf/icalp/GollerHOW10} shows that model
checking a fixed CTL formula on succinct one-counter automata 
(where counter increments and decrements are integers written in binary)
is EXPSPACE-hard.
Their proof is interesting and nontrivial, and uses two involved
results from complexity theory.
The technique of this proof was (a bit vaguely)  
referred to by Hunter~\cite{DBLP:conf/rp/Hunter15},
by which he derived EXPSPACE-hardness of reachability games 
on succinct one-counter nets
(with no zero tests).

Simulation-like equivalences on (non-succinct) one-counter nets 
are PSPACE-complete
(see~\cite{DBLP:journals/corr/HofmanLMT16} for simulation equivalence
and~\cite{DBLP:journals/jcss/BohmGJ14} for bisimulation equivalence). This
immediately yields EXPSPACE-upper bounds in the ``succinct'' cases.

The PSPACE---EXPSPACE gap for the simulation problem was also
mentioned in~\cite{DBLP:journals/corr/HofmanLMT16} from where we quote:
``Another direction for further research is to establish the exact
complexity of strong/weak simulation for OCN with binary encoded
increments and decrements on the counter. Trivially, the PSPACE-lower
bound applies for this model and an EXPSPACE upper bound follows from
the results of this paper with the observation that these more
expressive nets can be unfolded into ordinary OCN with an exponential
blow-up.''

Here, in this paper, we close the complexity gap by showing EXPSPACE-hardness
(and thus EXPSPACE-completeness),
using a defender-choice technique 
(cf., e.g.,~\cite{DBLP:journals/jacm/JancarS08}) 
to reduce reachability games to any relation between
simulation preorder and bisimulation equivalence.

But we first present a direct proof
of EXPSPACE-hardness of reachability games. 
(This makes our paper self-contained
and shows that we do not need to rely on the result 
from~\cite{DBLP:conf/icalp/GollerHOW10}.) 
Our  direct proof
is based on the technique
from~\cite{DBLP:journals/ipl/JancarS07} used there to show
PSPACE-hardness
for emptiness of alternating finite automata with one-letter
alphabet (thus giving an alternative proof for the result by
Holzer~\cite{DBLP:conf/dlt/Holzer95});
Srba~\cite{DBLP:journals/corr/abs-0901-2068} used this result 
to show PSPACE-hardness of behavioural relations for 
(non-succinct) one-counter nets.
In Section~\ref{sec:addrem} we discuss a relation of our proof to 
the countdown games of~\cite{DBLP:journals/lmcs/JurdzinskiSL08} and
their use in~\cite{DBLP:journals/ipl/Kiefer13}.
While the countdown games can serve as an interesting EXPTIME-complete problem,
by a slight enhancement we get an EXPSPACE-complete problem.

We stress that we do not provide an alternative proof for 
the EXPSPACE-hardness result in~\cite{DBLP:conf/icalp/GollerHOW10};
the result in~\cite{DBLP:conf/icalp/GollerHOW10} is stronger, 
though technically 
it does not induce the hardness results for reachability and
(bi)simulation games automatically. 

In Section~\ref{sec:overview} we give an informal overview 
of the ideas in our paper.
Section~\ref{sec:definitions} contains formal definitions,
Section~\ref{sec:reachgames} proves the EXPSPACE-hardness of
reachability games, and Section~\ref{sec:reachtosimul} reduces
reachability games to (bi)simulation games.
Section~\ref{sec:addrem} contains some additional remarks.

\section{Informal overview}\label{sec:overview}

The mentioned strong result in~\cite{DBLP:conf/icalp/GollerHOW10}
shows that for any fixed language $L$ in EXPSPACE we can for any word $w$
(in the alphabet of $L$) construct 
a succinct one-counter automaton that performs a computation which
is accepting iff $w\in L$.
Such a computation needs to access 
concrete bits in the (reversed) binary presentation of the counter value.
A straightforward direct access to such bits is destructive (the counter
value is lost after the bit is read) but this can be avoided:
instead of a ``destructive
reading'' the computation just ``guesses'' 
the respective bits, and it is forced to guess correctly by a carefully
constructed CTL formula that is required to be satisfied by the
computation.

If we imagine that there is an opponent who can challenge the guesses
about bits, and after a challenge a destructive test follows that
either confirms the guess or exposes its invalidity, then this readily leads
to the hardness results for reachability games, and then also for
behavioural relations by using a defender-choice technique (with two
synchronously evolving copies of the respective one-counter automaton).

Performing the above sketched procedure to prove the mentioned results
rigorously would require
recalling technical details from~\cite{DBLP:conf/icalp/GollerHOW10}.
Instead we give a direct self-contained proof, which also makes clear
that our results do not rely on the involved complexity results
used in~\cite{DBLP:conf/icalp/GollerHOW10}.

\subsubsection*{EXPSPACE-hardness of reachability games} 
We use a~``master'' reduction. We thus fix an arbitrary  language
$L$ in EXPSPACE, decided by a Turing machine $M$ in space $2^{p(n)}$
for a fixed polynomial $p$.
For any word $w$
in the alphabet of $L$ there is the respective computation of $M$,
which
is
accepting iff $w\in L$; the computation is a sequence
$C_0,C_1,\ldots,C_t$ of configurations, each $C_i$ being a string of
length $m=2^{p(\abs{w})}$. 

Any
$k\in\Nat$ can code the $j$-th position in $C_i$
where $i=k\div m$ ($\div$ is integer division) and $j=k\bmod m$ (assuming $i\leq
t$).
Given $w$, we can construct an alternating one-counter automaton,
with Eve's and Adam's control states one of which is Eve's winning
state; there are no zero-tests but
 transitions yielding negative counter values are not allowed.
Starting in the initial configuration $(p_0,0)$ (with zero in the
counter), Eve 
(who claims that there is an accepting computation
$C_0,C_1,\ldots,C_t$)
keeps
incrementing the counter until she enters a configuration
$(\langle q_+,a\rangle,k)$
by which she claims that the position coded by $k$ 
is the head-position in the accepting configuration $C_t$ and
contains letter $a$
($q_+$ being the accepting state of $M$).

Eve then subtracts $m$ and enters a control state corresponding to a triple of symbols that she claims 
to be the $(j{-}1)$-th, the $j$-th, and the $(j{+}1)$-th symbol in
$C_{t-1}$ (where $j=k\bmod m$); the triple must be consistent with
the current symbol $\langle q_+,a\rangle$ and the transition rules of $M$.
We note that at least one symbol (of the type $b$ or
$\langle q,b\rangle$) in the triple
must be incorrect when 
 $\langle q_+,a\rangle$ is incorrect (i.e., when it is 
 not really the $j$-th symbol in $C_t$).

Now Adam chooses a new current symbol from the triple,
and also adds ${-}1/0/{+}1$ to the counter accordingly.
Eve then presents
 another consistent triple, etc.

In fact, Eve can also present a pair instead of a triple, claiming
that $j=0$ or that $j=m{-}1$ (where $j=k\bmod m$). Adam can challenge
this, claiming that $j>0$ or that $j<m{-}1$;
similarly he can claim that $j=0$ or that $j=m{-}1$ when Eve provides
a triple. Such claims can be easily (destructively) verified:
We first let Eve decrement the counter by $m$ repeatedly. If she
leaves a number $k'\geq m$ in the counter, then Adam 
uses a transition subtracting $m$ that
enters a state precluding Eve's win.
Hence Eve is rather subtracting $m$ until 
 $j=k\bmod m$ is in the counter.
Similarly we implement the respective checks
of the claims $j=0$, $j=m{-}1$, $0<j<m{-}1$, so that Eve can force her
win precisely
when her claim was correct.

Finally, Eve can claim that $C_0$ has been reached (i.e., $k\div
m=0$), which can be again punished by Adam if not true.
In the case $k\div m=0$ Eve wins
if the control
state ``claims'' the $j$-th symbol of $C_0$ (for $w$),
which is the blank tape-symbol if $j\geq |w|$;
checking this 
 condition can be again
easily implemented in the game.

Hence $w\in L$ iff Eve can force reaching her winning
control state (when starting in $(p_0,0)$).

\subsubsection*{Reachability game reduces to (bi)simulation game}

Given a (succinct) one counter automaton with Eve's and Adam's control
states, one of them being Eve's winning, we first label each transition
by its unique action (action name) and 
take two copies of the resulting automaton 
(control state $s$ in one copy has a counterpart $s'$ in the other); 
in the first copy we add a
special (``winning'') action in Eve's winning state.

We let 
two players, called Attacker and Defender, to mimic the
reachability game. If there is Eve's turn, Attacker performs 
a transition in one copy, and Defender must do the same in the other
copy, being obliged to use the same action as Attacker (cf.
Fig~\ref{fig:attchoice}).

The defender-choice technique is used when there is Adam's turn. 
To this aim the two automaton-copies are a bit enhanced 
and interconnected
(as in Fig.~\ref{fig:defendchoice}, discussed later). By performing a
``choice-action'' $a_c$ Attacker lets Defender to choose from (more
than one)
transitions labelled with $a_c$; if Attacker 
does not follow Defender's choice in
the next round, then Defender installs syntactic equality
(in which case the play continues from a pair of 
the same configurations of the same automaton-copy).
It is thus Defender who chooses Adam's moves.

By the above construction we achieve that if Eve has a winning
strategy in the reachability game from $(p_0,0)$, then 
$(p_0,0)$ is not simulated by $(p'_0,0)$, and if 
Eve has no winning
strategy, then $(p_0,0)$ and $(p'_0,0)$ are bisimilar
(in the respective labelled transition system).

\begin{figure}
  \centering
  \begin{tikzpicture}[node distance=0.8cm]

    \node (dummy) {};
    \node[stav,left of=dummy,xshift=0cm] (s2) {$s_2$};
    \node[stav,right of=dummy,xshift=0cm] (s3) {$s_3$};
    \node[right of=dummy,xshift=4cm]  (dummy2) {};
    \node[stav,left of=dummy2] (s2prime) {$s'_2$};
    \node[stav,right of=dummy2] (s3prime) {$s'_3$};

  \node[stav,above of=dummy,yshift=1.5cm] (s1) {$s_1$}
  edge[pil] node[left] {$a^1_2(x)$} (s2.north)
  edge[pil] node[right] {$a^1_3(y)$} (s3.north);

  \node[stav,above of=dummy2,yshift=1.5cm] (s1prime) {$s'_1$}
  edge[pil] node[left] {$a^1_2(x)$} (s2prime.north)
  edge[pil] node[right] {$a^1_3(y)$} (s3prime.north);

  \node[stav,left of=s2,xshift=-5cm] (s2orig) {$s_2$};

  \node[stav,left of=s3,xshift=-5cm] (s3orig) {$s_3$};
  \node[stav,left of=s1,xshift=-5cm] (s1orig) {$s_1$}
  edge[pil] node[left] {$(x)$} (s2orig.north)
  edge[pil] node[right] {$(y)$} (s3orig.north);

  \node [left of=s1orig,xshift=0.2cm,yshift=0.2cm] (dummy2) {$E$};

  \node [right of=s1orig,xshift=2.3cm,yshift=0.5cm] (dummy3) {};
  \node [below of=dummy3,yshift=-2.8cm] (dummy4) {};

  \draw [thick] (dummy3) -- (dummy4) ;

 \end{tikzpicture}

	\caption{In $(s_1,s'_1)$ it is Attacker
	who chooses $(s_2,s'_2)$ or
	$(s_3,s'_3)$.
}\label{fig:attchoice}.
\end{figure}
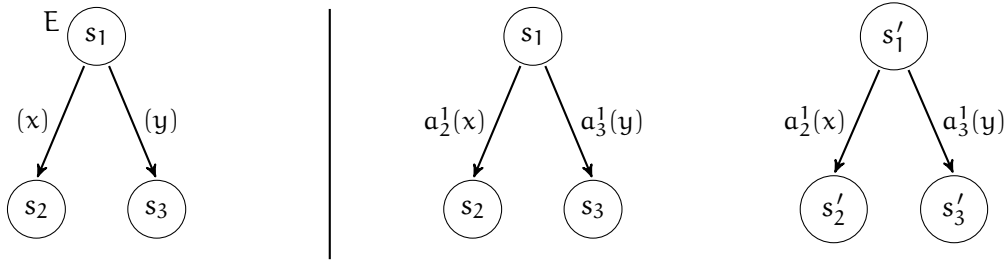

\begin{figure}
  \centering
  \begin{tikzpicture}[node distance=0.8cm]

    \node (dummy) {};
    \node[stav,left of=dummy,xshift=-0.4cm] (s2) {$s_2$};
    \node[stav,right of=dummy,xshift=0.4cm] (s3) {$s_3$};
   \node[stav,right of=s3,xshift=1.5cm] (s2prime) {$s'_2$};
\node[stav,right of=s2prime,xshift=1.5cm] (s3prime) {$s'_3$};

   \node[stav,above of=dummy,yshift=1.5cm] (s123) {$s^1_{23}$}
   edge[pil] node[left] {$a^1_2(x)$} (s2.north)
   edge[pil] node[left] {$a^1_3(y)$} (s3.north);

  \node[stav,above of=s2prime,yshift=1.5cm] (s12) {$s^1_{2}$}
  edge[pil] node[right] {$a^1_2(x'')$} (s2prime.north)
  edge[pil,bend right=25] node[right,xshift=-0.1cm,yshift=-0.1cm]
  {$a^1_3(y{-}x')$} (s3.north);

   \node[stav,above of=s3prime,yshift=1.5cm] (s13) {$s^1_{3}$}
   edge[pil, ] node[right] {$a^1_3(y'')$} (s3prime.north)
   edge[pil, out=-100, in=-60] node[right,yshift=-0.2cm]
   {$a^1_2(x{-}y')$} (s2.south);
 
  \node[stav,above of=s123,yshift=1.5cm] (s1) {$s_1$}
  edge[pil] node[left] {$a_c(0)$} (s123.north)
  edge[dsh, bend right=10] node[right,xshift=-0.4cm,yshift=0.5cm]
  {$a_c(x')$} (s12.west)
    edge[dsh, bend left=30] node[right,xshift=-0.9cm,yshift=0.5cm]
    {$a_c(y')$} (s13.west)
   ;

  \node[stav,above of=s12,yshift=1.5cm,xshift=1cm] (s1prime) {$s'_1$}
  edge[pil] node[left,xshift=-0cm,yshift=-0.1cm] {$a_c(x')$} (s12.north)
  edge[pil] node[right] {$a_c(y')$} (s13.north)
 ;

  \node[stav,left of=s2,xshift=-5cm] (s2orig) {$s_2$};

  \node[stav,left of=s3,xshift=-5cm] (s3orig) {$s_3$};
  \node[stav,left of=s1,xshift=-5cm] (s1orig) {$s_1$}
  edge[pil] node[left] {$(x)$} (s2orig.north)
  edge[pil] node[right] {$(y)$} (s3orig.north);

  \node [left of=s1orig,xshift=0.2cm,yshift=0.2cm] (dummy2) {$A$};

  \node [right of=s1orig,xshift=2.3cm,yshift=0.5cm] (dummy3) {};
  \node [below of=dummy3,yshift=-6.0cm] (dummy4) {};

  \draw [thick] (dummy3) -- (dummy4) ;

 \end{tikzpicture}

	\caption{In $(s_1,s'_1)$ it is Defender
	who chooses $(s_2,s'_2)$ or
	$(s_3,s'_3)$ (or a pair of equal states); to take 
	the counter-changes into account correctly, we put
	$x'=\min{\{x,0\}}$, $x''=\max{\{x,0\}}$, and
 $y'=\min{\{y,0\}}$, $y''=\max{\{y,0\}}$ 
 (hence $x=x'{+}x''$ and $y=y'{+}y''$).}\label{fig:defendchoice}.

\end{figure}
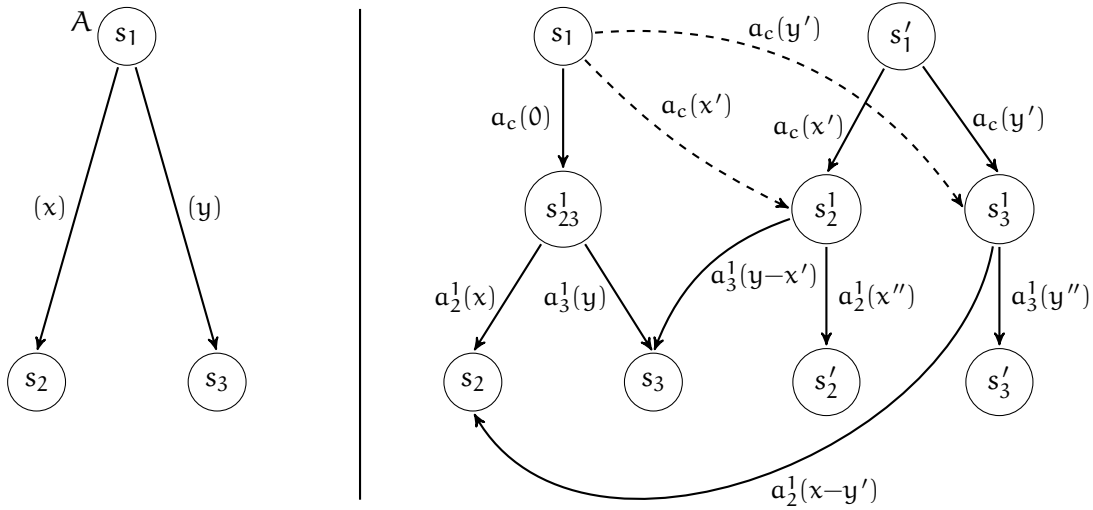

\section{Definitions}\label{sec:definitions}

By $\Zset$ we denote the set of integers,
and by $\Nat$ the set of nonnegative integers  
$\{0,1,2,\dots\}$.
\\
We use $[i,j]$ for denoting the set $\{i,i{+}1,\dots,j\}$, where $i,j\in\Zset$.

\paragraph{Reachability games.}
By a \emph{reachability game}, or an \emph{r-game} for short, we mean
a tuple
$\calG=(V,V_\exists, \gt{}, \calT)$,
where $V$ is the set of \emph{states} (or \emph{vertices}), 
$V_\exists\subseteq V$ is the set of \emph{Eve's states},
$\mathop{\gt{}}\subseteq {V\times V}$ is the \emph{transition relation}
(or the set of \emph{transitions}), and $\calT\subseteq V$ is the set
of \emph{target states}.
By  \emph{Adam's states} we mean  the elements of
$V_\forall=V\smallsetminus V_\exists$.

We put $\winareaE=\bigcup_{\lambda\in Ord}W_\lambda$ where $Ord$ is
the class of ordinals and the sets $W_\lambda\subseteq V$ are defined
inductively as follows. We put $W_0=\calT$. For $\lambda>0$ we put
$W_{<\lambda}=\bigcup_{\lambda'<\lambda}W_{\lambda'}$, and we
stipulate:
\begin{enumerate}[a)]
	\item
if $s\not\in W_{<\lambda}$, $s\in V_\exists$, and 
$s\gt{}\bar{s}$ for some $\bar{s}\in W_{<\lambda}$, then $s\in W_{\lambda}$;
	\item
if $s\not\in W_{<\lambda}$, $s\in V_\forall$, 
and we have $\emptyset\neq \{\bar{s}\mid s\gt{}\bar{s}\}\subseteq
W_{<\lambda}$,
then $s\in W_{\lambda}$.
\end{enumerate}		
(If a) applies, then $\lambda$ is surely a successor ordinal,
otherwise $\lambda$ can be also a limit ordinal.)

For each $s\in\winareaE$, by $\rank(s)$ we denote (the unique) 
$\lambda$ such that
$s\in W_\lambda$. A transition $s\gt{}\bar{s}$ is
\emph{rank-reducing} if $\rank({s})>\rank(\bar{s})$.
We note that for any $s\in\winareaE$ with $\rank(s)>0$ we have:
if $s\in V_{\exists}$, then there is at least one rank-reducing
transition $s\gt{}\bar{s}$ (in, fact $\rank(s)=\rank(\bar{s}){+}1$ in
this case); if $s\in V_{\forall}$, then there is at least one
transition $s\gt{}\bar{s}$ and all such transitions are rank-reducing.

\medskip

\emph{Remark.}
We are primarily interested in the games that have  (at most)
countably many states and are finitely branching (the sets 
$\{\bar{s}\mid s\gt{}\bar{s}\}$ are finite for all $s$). In such cases 
we have $\rank(s)\in\Nat$ for each $s\in\winareaE$.
\\
We also note that $\winareaE$ is the set of states from which Eve has a winning
strategy, i.e. such that guarantees reaching (some state in) $\calT$ when Eve
is choosing a next transition in Eve's states 
and Adam is choosing a next transition in Adam's states.

\paragraph{Labelled transition systems and (bi)simulations.}
A \emph{labelled transition system}, an \emph{LTS} for short, is a
tuple
$\calL = (S, Act, (\gt{a})_{a\in Act})$
  where $S$ is the set of \emph{states}, $Act$ is the set of
  \emph{actions},
  and $\mathop{\gt{a}}\subseteq S \times S$ is the set
of  \emph{$a$-transitions} (transitions labelled with $a$), for each
  $a\in Act$.

Given $\calL = (S, Act, (\gt{a})_{a\in Act})$,
a relation $R \subseteq S \times S$ is a \emph{simulation}
  if for every $(s,s')\in R$ and every
  $s \gta t$ there is $s' \gta t'$ such that $(t,t')\in R$;
  if, moreover, for every $(s,s')\in R$ and every
  $s' \gta t'$ there is $s \gta t$ such that $(t,t')\in R$, then $R$
is a \emph{bisimulation}.
  The union of all simulations (on $\calS$) is the maximal
  simulation, denoted $\simul$; it is a preorder, called
  \emph{simulation preorder}.
 The union of all bisimulations is the maximal
 bisimulation, denoted $\bisim$; it is an equivalence, called
 \emph{bisimulation equivalence} (or \emph{bisimilarity}). 
 We observe that
 $\bisim\subseteq\simul$.

 \medskip

 \emph{Remark.} 
We can write
 $s_1\simul s_2$ or $s_1\bisim s_2$ also for states $s_1$,
 $s_2$ from different  LTSs $\calL_1$, $\calL_2$, in which case
 the LTS arising by the disjoint union of $\calL_1$ and $\calL_2$ is
 (implicitly) referred to.
\\
It is also useful to think in terms of games here.
In the \emph{simulation game}, in a (current) pair $(s,s')$ Attacker chooses 
a transition $s\gt{a}t$ and Defender responds with some
$s'\gt{a}t'$ (for the action $a$ chosen by Attacker); the play then
continues with another round, now in the current pair  $(t,t')$, etc.
If Defender has no response in a round, then Attacker wins the play.
It is standard to note that $s\notsimul s'$ iff Attacker has a winning
strategy from $(s,s')$.
\\
The case of \emph{bisimulation game} is analogous, but 
in any round starting from  $(s,s')$
Attacker can choose to
play  $s\gt{a}t$ or $s'\gt{a}t'$, and Defender has to respond
with some $s'\gt{a}t'$ or $s\gt{a}t$, respectively.
Now Attacker has a winning
strategy from $(s,s')$ iff $s\notbisim s'$.

\medskip

We now define specific r-games and LTSs, presented by particular
one-counter automata.

\paragraph{(Succinct) one-counter net games.}
By a \emph{one-counter net game}, an \emph{ocn-game} for short, we
mean a tuple $\calN=(Q,Q_{\exists},\delta,p_{win})$
where $Q$ is the finite set of \emph{(control) states},
$Q_{\exists}\subseteq Q$ is the set of \emph{Eve's (control) states}, 
$p_{win}\in Q$ is the \emph{target (control) state}, and 
$\delta\subseteq Q\times\Zset\times Q$ is the finite set of
\emph{(transition) rules}.
We often present a rule $(q,z,q')\in\delta$ as
$q\gt{z}q'$.
By \emph{Adam's (control) states} we   
mean the elements of $Q_{\forall}=Q\smallsetminus Q_{\exists}$.

An ocn-game $\calN=(Q,Q_{\exists},\delta,p_{win})$ has the
\emph{associated r-game} 
\begin{equation}\label{eq:gamefornet}
\calG_\calN=(Q\times
\Nat,Q_{\exists}\times\Nat,\gt{},\{p_{win}\}\times\Nat)
\end{equation}
where $(q,m)\gt{}(q',n)$ iff $q\gtl{n-m}q'$ is a rule (in $\delta$).
We often write $q(m)$ instead of $(q,m)$ for states of $\calG_\calN$.
(A rule  $q\gt{z}q'$ thus induces transitions $q(m)\gt{}q'(m{+}z)$ for
all $m\geq \max\{0,{-}z\}$.)

\medskip

We define the problem \emph{RG-SOCN} (reachability game on succinct
one-counter nets):
\begin{quote}
\emph{Instance}: an ocn-game $\calN$
with integers $z$ in rules $q\gt{z}q'$ written in binary, 
\\
\hspace*{4.5em}and a control state $p_0$.
\\
\emph{Question}: Is $p_0(0)\in\winareaE$ in the game $\calG_\calN$~?
\end{quote}

\medskip

\emph{Remark.}
We have defined the target states (in $\calG_\calN$) by the control
state $p_{win}$. There are other natural variants
(e.g., one in~\cite{DBLP:conf/rp/Hunter15} defines the target set
$\{p(0)\mid p\neq p_0\}$) that are, in principle, equivalent in our
context.

\paragraph{(Succinct) labelled one-counter nets.}
A \emph{labelled one-counter net}, an \emph{OCN} for short,
is a triple $\calN=(Q,Act,\delta)$,
where $Q$ is the finite set of \emph{control states}, $Act$ the finite
set of \emph{actions}, 
and $\delta \subseteq Q \times Act\times \Zset \times Q$ is 
the finite set of (\emph{labelled transition}) \emph{rules}.
We present a rule $(q,a,z,q')\in\delta$ as
$q\gtl{a,z}q'$ ($a\in Act$, $z\in\Zset$).
An OCN $\calN=(Q,Act,\delta)$  has the
\emph{associated LTS} 
\begin{equation}\label{eq:ltsfornet}
\calL_\calN=(Q\times
\Nat,Act,(\gt{a})_{a\in Act})
\end{equation}
where $q(m)\gtl{a}q'(n)$ iff $q\gtl{a,n-m}q'$ is a rule in $\delta$.
(We again write $q(m)$ instead of $(q,m)$.
A rule  $q\gtl{a,z}q'$ thus induces transitions $q(m)\gt{a}q'(m{+}z)$ for
all $m\geq \max\{0,{-}z\}$.)

For claims on complexity (in particular for
Theorem~\ref{th:equivhard}) we define 
a \emph{succinct (labelled) one-counter net}, 
a \emph{SOCN} for short, as an OCN where the integers in 
rules $q\gtl{a,z}q'$ are written in binary.

\section{EXPSPACE-hardness of reachability games}\label{sec:reachgames}

In Section~\ref{sec:overview} we sketched a ``master'' reduction
showing EXPSPACE-hardness of reachability games on succinct
one-counter nets, which is captured by the following theorem. 

\begin{theorem}\label{th:rgexpsphard}
The problem \rgsocn is EXPSPACE-hard.
\end{theorem}

As already mentioned, an idea of a proof is given 
in~\cite{DBLP:conf/rp/Hunter15} by referring
to~\cite{DBLP:conf/icalp/GollerHOW10}. Here we provide a
self-contained proof, by performing the sketched master reduction  in
detail.

We first give a construction for general Turing machines,
and then implement its space-bounded variant by one-counter nets.
Hence we now fix an arbitrary (deterministic) 
Turing machine~$M=(Q,\Sigma,\Gamma,\delta,q_0,\{q_+,q_-\})$,
where $Q$ is the set of (control) states, $q_0\in Q$ the initial
state,  $q_+\in Q$ the accepting
state,  $q_-\in Q$ the rejecting
state, 
$\Sigma$ the input
alphabet, $\Gamma\supseteq \Sigma$ the tape alphabet, satisfying 
$\blank\in\Gamma\smallsetminus\Sigma$ for the special blank
tape-symbol $\blank$, and $\delta: (Q\smallsetminus
\{q_+,q_-\})\times\Gamma\rightarrow
Q\times\Gamma\times\{{-}1,{+}1\}$ is the transition function.

Putting $\Delta=\Gamma\cup(Q\times \Gamma)$, we define the relation
$\vdash\mathop{\subseteq} {\Delta^3\times \Delta}$ 
in a standard way:
\\
$(\beta_1,\beta_2,\beta_3)\vdash \beta$ if $\beta_i\in Q\times\Gamma$
for at most one $i\in\{1,2,3\}$ and the following conditions
hold:
\begin{itemize}
\item if $\beta_1\beta_2\beta_3=(q,x)yz$ and 
	$\delta(q,x)=(q',x',d)$, then $\beta=(q',y)$ if $d={+}1$ and
	$\beta=y$ otherwise (i.e., if $d={-}1$);
\item if $\beta_1\beta_2\beta_3=x(q,y)z$ and 
	$\delta(q,y)=(q',y',d)$, then $\beta=y'$ (for any
	$d\in\{{-}1,{+}1\}$);
\item if $\beta_1\beta_2\beta_3=xy(q,z)$ and 
	$\delta(q,z)=(q',z',d)$, then $\beta=(q',y)$ if $d={-}1$ and
	$\beta=y$ otherwise;
\item if  $\beta_1\beta_2\beta_3=xyz$, then $\beta=y$. 
\end{itemize}
We note that $\vdash$ is a partial function, in fact.
By a \emph{configuration} of $M$ we mean a mapping
$C:\Zset\rightarrow\Delta$
where $C(j)\neq \blank$ for only finitely many $j\in\Zset$ and
$C(j)\in Q\times \Gamma$ for precisely one $j\in\Zset$, called the
\emph{head-position}; if  $C(j)=(q_+,x)$  for the head-position $j$ 
(and $x\in\Gamma$)
then $C$
is \emph{accepting}, and if  $C(j)=(q_-,x)$  then $C$
is \emph{rejecting}.

We put $C\vdash C'$ (thus overloading the
symbol $\vdash$) if $\big(C(j{-}1), C(j), C(j{+}1)\big)\vdash C'(j)$ for all
$j\in\Zset$. This relation $\vdash$ is again a partial function; if
$C$ is \emph{final}, i.e. accepting or rejecting, then there is no $C'$ such that $C\vdash
C'$.

Given a word $w=a_1a_2\cdots a_n\in \Sigma^*$ (hence $|w|=n$), we define the
respective initial configuration as
$C_0^w$ where $C_0^w(0)=(q_0,a_1)$ if $n\geq 1$ 
and $C_0^w(0)=(q_0,\blank)$ if $n=0$, $C^w_0(j)=a_{j+1}$ for all
$j\in[1,n{-}1]$, and $C^w_0(j)=\blank$ for all $j<0$ and all $j\geq
n$. If $C^w_i$ is not final,
then we define
$C^w_{i+1}$ so that $C^w_i\vdash C^w_{i+1}$.
The \emph{computation on} $w$ is either the finite sequence 
$C_0^w,C^w_1,C^w_2,\dots, C_t^w$ where  $C_t^w$ is final (accepting or
rejecting), or the infinite sequence $C_0^w,C^w_1,C^w_2,\dots$; 
formally we put $C^w_i(j)=\bot$ (for $\bot\not\in\Delta$) if
 there is a final $C^w_t$ and $i>t$.

By $L(M)$ we denote the \emph{language accepted by} $M$,
i.e. the set $\{w\in\Sigma^*\mid$ the computation on $w$ finishes with
an accepting configuration$\}$.

Given (our fixed) Turing machine $M$ and a word $w=a_1a_2\dots
a_n\in\Sigma^*$,
we define the r-game 
$$\calG^M_w=(V,V_\exists, \gt{}, \calT)$$
where $V_\exists=\{s_0,s_{F}\}\cup(\Delta\times\Nat\times\Zset)$,
$V_\forall=V\smallsetminus V_{\exists}=\Delta^3\times\Nat\times\Zset$, $\calT=\{s_F\}$,
and where the transition
relation $\gt{}$ is defined as
follows:
\begin{enumerate}[1)]
	\item		Eve's moves:
		\begin{enumerate}[a)]
			\item 
				$s_0\gt{}((q_+,x),i,j)$ for all
				$x\in\Gamma, i\in\Nat,j\in\Zset$;
			\item
		$(\beta,i,j)\gt{}((\beta_1,\beta_2,\beta_3), i{-}1,j)$
if $i\geq 1$ and 
$(\beta_1,\beta_2,\beta_3)\vdash \beta$;
\item
$(\beta,0,j)\gt{}s_F$ if $\beta=C^w_0(j)$.
\end{enumerate}
	
\item Adam's moves:
$((\beta_1,\beta_2,\beta_3),i,j)\gt{}(\beta_\ell,i,j{-}2{+}\ell)$
for  $\ell\in\{1,2,3\}$.
\end{enumerate}
The next proposition shows how Eve's winning region is related to
$M$'s computation on $w$.

\begin{prop}\label{prop:abstractgame}
In $\calG^M_w$ we have
$\winareaE=X_0\cup X_1\cup X_2$
where 
\begin{itemize}
	\item
		$X_0=\{s_0,s_F\}$ if $w\in L(M)$ and $X_0=\{s_F\}$ if  $w\not\in
L(M)$,
\item
$X_1=\big\{(\beta,i,j)\mid C^w_i(j)=\beta\big\}$,
\item
$X_2=\big\{((\beta_1,\beta_2,\beta_3),i,j)\mid 
C^w_i(j{-}1)C^w_{i}(j)C^w_i(j{+}1)=\beta_1\beta_2\beta_3\big\}$.
\end{itemize}
\end{prop}
\begin{proof}
By definition we have $s_F\in\winareaE$. We first ignore the question
if $s_0\in\winareaE$, and show the rest of the claim by induction 
on $i\in\Nat$.
The base case ($i=0$) is clear due to the transitions from the points
1c) and 2).
The induction step is also easy to check, when we note 
that for any 
$(\beta,i,j)$ with $i\geq 1$ we have:
\begin{itemize}
	\item
		if $\beta=C^w_i(j)$ then there is a transition
		$(\beta,i,j)\gt{}((\beta_1,\beta_2,\beta_3), i{-}1,j)$
		such that 
		\\
		$\beta_1\beta_2\beta_3=C^w_{i-1}(j{-}1)C^w_{i-1}(j)C^w_{i-1}(j{+}1)$, and
	\item
if $\beta\neq C^w_i(j)$ then for every  transition
		$(\beta,i,j)\gt{}((\beta_1,\beta_2,\beta_3), i{-}1,j)$
		there is $\ell\in\{1,2,3\}$ such that 
		$\beta_\ell\neq C^w_{i-1}(j{-}2{+}\ell)$.
\end{itemize}
The claim for $s_0$ follows, since 
 $w\in L(M)$ iff 
$C^w_t(j)=(q_+,x)$ for some $t\in\Nat$, $j\in\Zset$,  $x\in\Gamma$.
\end{proof}	

If we think of mimicking the game $\calG^M_w$ by a one-counter game, it
is natural to represent a state $(\beta,i,j)$, or
$((\beta_1,\beta_2,\beta_3),i,j)$, so that $\beta$, or $(\beta_1,\beta_2,\beta_3)$, is (in)
a control state and $(i,j)$ is suitably represented by one (counter)
value $k$.
This seems manageable when we are guaranteed
that 
in the computation of $M$ on
$w=a_1a_2\cdots a_n$ 
the head-position is never outside $[0,m{-}1]$ for a fixed $m\geq
n$; we now assume this, while also assuming $n\geq 1$ for convenience.
In such a case $\calG^M_w$ can be naturally
adjusted to yield a game $\calG^M_{w,m}$ in which
the head-position is kept inside
$[0,m{-}1]$. Our aim is to mimic $\calG^M_{w,m}$ by 
(the r-game associated with) an ocn-game $\calN^M_{w,m}$; a first attempt 
to built such  $\calN^M_{w,m}$ can look as follows:

\begin{equation}\label{eq:ocgamePone}
	\begin{array}{lllll}
	   p_0 \trans{+1} p_0 &
	   p_0 \trans{0} s_{(q_+,x)} & (\textnormal{where } x\in \Gamma)\\
		s_\beta\trans{-m}\AStateIII & & (\textnormal{where }(\beta_1,\beta_2,\beta_3)\vdash\beta)\\
	   \AStateIII \trans{-1} s_{\beta_1} &  \AStateIII \trans{0} s_{\beta_2} 
	   &\AStateIII \trans{+1}s_{\beta_3}\\    
   
   \end{array}
\end{equation}
By the configuration $p_0(0)$ we represent the state $s_0$ of
$\calG^M_{w,m}$.
A configuration $s_\beta(k)$ is intended to represent the state
$(\beta, k\div m, k\bmod m)$; similarly
$s_{(\beta_1,\beta_2,\beta_3)}(k)$ is intended to represent the state
$((\beta_1,\beta_2,\beta_3), k\div m, k\bmod m)$.
Hence $p_0$ and $s_\beta$ are Eve's control states, while 
$s_{(\beta_1,\beta_2,\beta_3)}$ are Adam's control states.

The ocn-game given by the rules in~(\ref{eq:ocgamePone})
does not mimic the game $\calG^M_{w,m}$ faithfully,
due to possible ``cheating'' related to the boundary
head-positions.
Therefore we add 

\begin{equation}\label{eq:ocgamePtwo}
	\begin{array}{lllll}
		s_\beta\trans{-m}s^L_{(\blank,\beta_2,\beta_3)} 
		&		s^L_{(\blank,\beta_2,\beta_3)}\trans{0}\bar{g}_0
		&
		s^L_{(\blank,\beta_2,\beta_3)}\trans{0}s_{\beta_2}
		&
s^L_{(\blank,\beta_2,\beta_3)}\trans{+1}s_{\beta_3}
   \end{array}
\end{equation}
for the cases $(\blank,\beta_2,\beta_3)\vdash\beta$.
By entering a configuration $s^L_{(\blank,\beta_2,\beta_3)}(k)$
Eve also ``claims'' that the head-position is $0$, i.e., that $k\bmod m=0$.
The state $s^L_{(\blank,\beta_2,\beta_3)}$ is Adam's, who can believe
the claim and play accordingly, or
decide to challenge the claim by performing  
$s^L_{(\blank,\beta_2,\beta_3)}(k)\gt{0}\bar{g}_0(k)$.
Symmetrically we add 
\begin{equation}\label{eq:ocgamePthree}
	\begin{array}{lllll}
		s_\beta\trans{-m}s^R_{(\beta_1,\beta_2,\blank)} 
		&
		s^R_{(\beta_1,\beta_2,\blank)}\trans{-1}s_{\beta_1}
				&
		s^R_{(\beta_1,\beta_2,\blank)}\trans{0}s_{\beta_2}
		&
		s^R_{(\beta_1,\beta_2,\blank)}\trans{0}\bar{g}_{m-1}
   \end{array}
\end{equation}
for the cases $(\beta_1,\beta_2,\blank)\vdash\beta$.
By entering a configuration $s^R_{(\beta_1,\beta_2,\blank)}(k)$
Eve also claims that the head-position is $m{-}1$, i.e., that $k\bmod
m=m{-}1$.
In a state $s^R_{(\beta_1,\beta_2,\blank)}$ Adam can
decide to challenge the claim by performing  
$s^R_{(\beta_1,\beta_2,\blank)}(k)\gt{0}\bar{g}_{m-1}(k)$.

To complete this reasoning, by entering 
$\AStateIII$ (with no superscript) Eve claims that
$0<(k\bmod m)<m{-}1$; Adam can challenge this by using one of the rules
\begin{equation}\label{eq:ocgamePfour}
	\begin{array}{lllll}
	\AStateIII\trans{0}{g}_{0}
	& & \AStateIII\trans{0}{g}_{m-1}
   \end{array}
\end{equation}
Hence control states $g_c$, for $c\in\{0,m{-}1\}$, can be viewed as
Adam's claims
``the current counter value $k$ satisfies $k\bmod m=c$''; similarly 
$\bar{g}_c$ is Adam's claim
``$k\bmod m\neq c$''. We need to add some rules guaranteeing Eve's win 
in the configurations $g_c(k)$ and $\bar{g}_c(k)$ precisely when
Adam's claims are incorrect. Moreover, we need that in $s_\beta(k)$
where $k<m$ Eve can force her win iff $\beta=C^w_0(k)$.
The required properties are achieved by completing 
the rules~(\ref{eq:ocgamePone}),
(\ref{eq:ocgamePtwo}),~(\ref{eq:ocgamePthree}),~(\ref{eq:ocgamePfour})
with
\begin{equation}\label{eq:ocgamePfive}
   \begin{array}{lllllllll}
	   s_{\blank}\trans{0}f  & &  & s_\beta\trans{-j}\bar{e}_0 
	   & \textnormal{where } 0\leq j\leq n{-}1 \textnormal{ and }
	   C^w_0(j)=\beta
   \end{array}
\end{equation}
and with the following final set:
\begin{equation}\label{eq:ocgamePsix}
   \begin{array}{lllllllll}
	   f \trans{-m} \WinA & e_c \trans{-c} e_c'  &
	   \bar{e}_c \trans{-c} \bar{e}_c' &
	   g_c \trans{-m} g_c & \bar{g}_c \trans{-m} \bar{g}_c 
       \\
       f \trans{0} f' &
   e_c \trans{0} \WinE  &  \bar{e}_c' \trans{-1} \WinA & 
      g_c \trans{0} g_c' &  
      
       \bar{g}_c \trans{0} \bar{g}_c'       \\
       f' \trans{-n} \WinE &
   e_c' \trans{-1} \WinE 
       &     \bar{e}_c' \trans{0} \WinE       &
      g_c' \trans{-m} \WinE &  \bar{g}_c' \trans{-m} \WinE 
     \\
     & &  & 
      g_c' \trans{0} e_c &  \bar{g}_c' \trans{0} \bar{e}_c &&
      \\
   \end{array}
\end{equation}
Here $c$ ranges over two values, $0$ and $m{-}1$. 

Given $M$, $w=a_1a_2\cdots a_n$, and $m\geq n$,
we thus constructed the ocn-game
\begin{equation}\label{eq:constrocngame}
\calN^M_{w,m}=(Q,Q_{\exists}, \delta_{\calN}, \WinE)
\end{equation}
where
\begin{itemize}
	\item		
$Q_{\exists}=\{p_0, f',e_c', \bar{e}_c, g_c', \bar{g}'_c, \WinE, \WinA\}\cup
\{s_\beta\mid \beta\in\Delta\}$, 
\item
$Q_{\forall}=\{f, e_c, \bar{e}'_c, g_c, \bar{g}_c\}\cup
\{s_{(\beta_1,\beta_2,\beta_3)}\mid \beta_i\in\Delta\}\cup
\{s^L_{(\blank,\beta_2,\beta_3)}\mid\beta_i\in\Delta\}\cup
\{s^R_{(\beta_1,\beta_2,\blank)}\mid\beta_i\in\Delta\}$
\end{itemize}
(recall that 	$Q_{\forall}=Q\smallsetminus Q_{\exists}$),
and the set $\delta_\calN$ of rules is given 
by~(\ref{eq:ocgamePone}), (\ref{eq:ocgamePtwo}), (\ref{eq:ocgamePthree}),
(\ref{eq:ocgamePfour}), (\ref{eq:ocgamePfive}), (\ref{eq:ocgamePsix}).

\begin{prop}\label{prop:spaceboundedocgame}
In the r-game associated with $\calN^M_{w,m}$ we have $p_0(0)\in\winareaE$ iff
the computation of $M$ on $w$ never moves the head out 
of $[0,m{-}1]$ and finishes with accepting.
\end{prop}
\begin{proof}
The claim follows by Prop.~\ref{prop:abstractgame}, by the
previous discussions accompanying the construction of
$\calN^M_{w,m}$, and by the following properties that are easy to check
(recall that $c\in\{0,m{-}1\}$):
\begin{enumerate}[1.]
	\item		$f(k)\in \winareaE$ iff $n\leq k<m$;
	\item 		$e_c(k)\in\winareaE$ iff $k\neq c$, and 
$\bar{e}_c(k)\in\winareaE$ iff $k=c$;
	\item 		$g_c(k)\in\winareaE$ iff $k\bmod m\neq c$, and 
$\bar{g}_c(k)\in\winareaE$ iff $k\bmod m=c$;
\item
	for $k<n$, $s_{\beta}(k)\in\winareaE$ iff $\beta=C^w_0(k)$.
\end{enumerate}		
Using the rules~(\ref{eq:ocgamePfive}) in 
$s_{\beta}(k)$ when  $k\geq m$ is thus losing for Eve; by this
the proof is finished.
\end{proof}

We now note that the control states
in $\calN^M_{w,m}$ are determined by $M$, as well as
the rules  except of those in~(\ref{eq:ocgamePfive})
that are dependent on $w$; to be precise, the values of 
counter-decrements ${-}m$,
${-(m{-}1)}$ (in ${-}c$ for $c=m{-}1$) and ${-}n$ are not determined
by $M$ but by ``parameters'' $w$ and $m$.

To finish the proof of Theorem~\ref{th:rgexpsphard}, we assume an
arbitrary
fixed language $L$ in \EXPSPACE. There is thus a Turing machine $M$ and a
 polynomial $p$ such that $M$ decides $L$ and 
the head-position in the computation of $M$ on any $w$
(in the alphabet of $L$)
never moves out of the
interval $[0,m{-}1]$ where $m=2^{p(n)}$ for $n=|w|$. 
Given $w$, it is straightforward to construct $\calN^M_{w,m}$,
by filling the rules~(\ref{eq:ocgamePfive})
and the parameters $n,m$ into a fixed scheme.
Since $m$ can be presented in binary by using $p(n){+}1$ bits, we can
construct $\calN^M_{w,m}$ 
in
logarithmic work-space (from a given $w$).

\section{Reducing reachability games to
(bi)simulation games}\label{sec:reachtosimul}

We first discuss a reduction in a general framework,
and then apply it to the case of (succinct) one-counter nets.

We assume a (general) r-game
$\calG=(V,V_\exists, \gt{}, \calT)$, and below we 
define the LTS 
\begin{equation}\label{eq:ltstogame}
\calL(\calG)=(S,Act,(\gt{a})_{a\in Act}).
\end{equation}
(Cf. Fig.~\ref{fig:attchoice} and~\ref{fig:defendchoice}, where we now
ignore the bracketed parts of transition-labels.)

\medskip
\noindent
The set $S$ is defined as follows (recall that
$V_{\forall}=V\smallsetminus V_{\exists}$): 
\begin{itemize}
	\item every $s\in V$ and its
		``copy'' $s'$ is in $S$;
	\item
		if $s\in V_{\forall}$ and $s\gt{} \bar{s}$, then a state
		$\langle s,\bar{s}\rangle$ is in $S$ 
		\\
		(in Fig.~\ref{fig:defendchoice}
we write, e.g., $s^1_3$ instead of 
$\langle s_1,s_3\rangle$);		
	\item
if $s\in V_{\forall}$ and $X=\{\bar{s}\mid s\gt{}	\bar{s}\}$ is
nonempty,
then a state $\langle s,X\rangle$ is in $S$ 
\\		
(in Fig.~\ref{fig:defendchoice}
		we write $s^1_{23}$ instead of 
		$\langle s_1,\{s_2,s_3\}\rangle$).
\end{itemize}
We put  $Act=\{a_c,a_{win}\}\cup\{a_{\langle s,\bar{s}\rangle}\mid
s\gt{}\bar{s}\}$ and define $\gt{a}$ for $a\in Act$ as follows: 
\begin{itemize}
	\item if $s\in V_{\exists}$ and $s\gt{} \bar{s}$, then  
		$s\gtl{a_{\langle s,\bar{s}\rangle}} \bar{s}$ and
		$s'\gtl{a_{\langle s,\bar{s}\rangle}} \bar{s}'$
\\
		(in Fig.~\ref{fig:attchoice} we write, e.g., $a^1_3$ instead of
	$a_{\langle s_1,s_3\rangle}$); 
\item
 if $s\in V_{\forall}$ and 
 $X=\{\bar{s}\mid s\gt{}\bar{s}\}\neq\emptyset$,
 then
	\begin{enumerate}[a)]
		\item
		$s\gtl{a_c}\langle s,X\rangle$,
		and $s\gtl{a_c} \langle s,\bar{s}\rangle$,
		$s'\gtl{a_c} \langle s,\bar{s}\rangle$
		for all $\bar{s}\in X$
\\
	(cf. Fig.~\ref{fig:defendchoice} 
	where $s=s_1$ and $X=\{s_2,s_3\}$
	and consider dashed edges as
	normal edges; the subscript $c$ in $a_c$ stands for
	``choice'');
\item
	for each $\bar{s}\in X$ we have
$\langle s,X\rangle\gtl{a_{\langle
s,\bar{s}\rangle}}\bar{s}$
and $\langle s,\bar{s}\rangle\gtl{a_{\langle
s,\bar{s}\rangle}}\bar{s}'$;
moreover, for each $\bar{\bar{s}}\in X\smallsetminus \{\bar{s}\}$ 
we have $\langle s,\bar{s}\rangle\gtl{a_{\langle
s,\bar{\bar{s}}\rangle}}\bar{\bar{s}}$
\\
(e.g., in Fig.~\ref{fig:defendchoice} we thus have
$s^1_2\gtl{a^1_2}s'_2$ and $s^1_2\gtl{a^1_3}s_3$). 
\end{enumerate}
\item
for each $s\in\calT$ we have $s\gtl{a_{win}} s$ 
(for special $a_{win}$ that is not performable from $s'$).
\end{itemize}

We recall that $\bisim{\subseteq}\simul$ where
$\simul$ denotes simulation preorder and
$\bisim$ bisimulation equivalence.

\begin{prop}\label{prop:reachreduced}
	For any $s\in V$ and any relation
	$\rho$ satisfying
	$\bisim{\subseteq}\mathop{\rho}{\subseteq}\simul$
	we have: 
	\begin{enumerate}[a)]
		\item
if	$s\in\winareaE$ (in $\calG$), then $s\notsimul s'$ (in
$\calL(\calG)$) and thus $(s,s')\not\in \rho$;
\item
if $s\not\in\winareaE$, then $s\bisim s'$ and thus $(s,s')\in \rho$.	
 \end{enumerate}
\end{prop}
\begin{proof}
a) For the sake of contradiction suppose that there is $s\in\winareaE$ 
such that $s\simul s'$; we consider such $s$ with 
the least rank.
We note that $\rank(s)>0$, 
since $s\in\calT$ entails $s\notsimul s'$ due to the
	transition $s\gtl{a_{win}} s$.

If $s\in V_{\exists}$, then let $s\gt{} \bar{s}$ be a rank-reducing
transition. Attacker's move $s\gtl{a_{\langle s,\bar{s}\rangle}}
\bar{s}$, from the pair $(s,s')$,
must be responded with $s'\gtl{a_{\langle s,\bar{s}\rangle}} \bar{s}'$;
but we have $\bar{s}\notsimul \bar{s}'$ by the ``least-rank''
assumption,
which contradicts with the assumption $s\simul s'$.

If $s\in V_{\forall}$, then $X=\{\bar{s}\mid s\gt{}\bar{s}\}$ is
nonempty (since $s\in\winareaE$) and 
$\rank(\bar{s})<\rank(s)$ for all $\bar{s}\in X$.
For the pair  $(s,s')$ 
we now consider Attacker's move
$s\gtl{a_{c}} \langle s,X \rangle$. Defender can choose 
$s'\gtl{a_{c}} \langle s,\bar{s} \rangle$ for any $\bar{s}\in X$
(recall that $\rank(\bar{s})<\rank(s)$).
In the current pair $(\langle s,X \rangle, \langle s,\bar{s} \rangle)$
Attacker can play $\langle s,X \rangle \gtl{a_{\langle s,\bar{s}
\rangle}} \bar{s}$, and this must be responded by 
$\langle s,\bar{s} \rangle \gtl{a_{\langle s,\bar{s}
\rangle}} \bar{s}'$. But we again have 
$\bar{s}\notsimul \bar{s}'$ by the ``least-rank'' assumption,
which contradicts with $s\simul s'$.

\medskip

b) It is easy to verify that the following set
is a bisimulation in $\calL_\calG$:
\\
 $\{(s,s)\mid s\in S\}
 \cup\{(s,s')\mid s\in V\smallsetminus \winareaE\}
 \cup \{(\langle s,X \rangle, \langle s,\bar{s} \rangle)\mid 
 s\in V_{\forall}\smallsetminus \winareaE, \bar{s}\in V\smallsetminus
 \winareaE\}$.
\end{proof}	

\emph{Remark.}
We can note that the dashed edges in Fig.~\ref{fig:defendchoice} are
not necessary for the simulation game (i.e., if b) in
Prop.~\ref{prop:reachreduced} is reformulated to 
``if $s\not\in\winareaE$, then $s\simul s'$'');
 they are important for the
bisimulation game.

\medskip

Now we apply the described reduction to succinct one-counter nets to
obtain:

\begin{theorem}\label{th:equivhard}
	For succinct labelled one-counter nets (SOCNs),
deciding any relation containing bisimulation equivalence and 
contained in simulation preorder is EXPSPACE-hard. 
\end{theorem}
\begin{proof}
Let us fix a language $L$ in EXPSPACE,
and a Turing machine $M$ and a
 polynomial $p$ such that $M$ decides $L$ and 
the head-position in the computation of $M$ on any $w$
never moves out of
$[0,m{-}1]$ where $m=2^{p(|w|)}$. 

Given $w$, we can construct $\calN^M_{w,m}$
as defined by~(\ref{eq:constrocngame}) 
in Section~\ref{sec:reachgames} (in logarithmic work-space).
By Prop.~\ref{prop:spaceboundedocgame} we have $w\in L(M)$ 
iff $p_0(0)\in \winareaE$ in the game $\calG$ associated with
$\calN^M_{w,m}$ (i.e., in $\calG_\calN$ defined
by~(\ref{eq:gamefornet}) for $\calN=\calN^M_{w,m}$), hence iff 
$\big(p_0(0), (p_0(0))'\big)\not\in\rho$ in the LTS $\calL(\calG)$
for any $\rho$ satisfying $\bisim\subseteq\rho\subseteq\simul$
(by Prop.~\ref{prop:reachreduced}).

Therefore we will be done once we show the following claim.

\medskip

\noindent
\emph{Claim}.
For any succinct ocn-game $\calN=(Q,Q_{\exists}, \delta, \WinE)$ 
(where ``succinct'' refers to binary presentations of $z$ in $q\gt{z}q'$)
we can construct, in logarithmic work-space, a SOCN $\calN'$ such that 
the LTS $\calL_{\calN'}$ (as defined by~(\ref{eq:ltsfornet})) is
isomorphic with $\calL(\calG_\calN)$.

\medskip

A proof of this claim is also demonstrated in Figures~\ref{fig:attchoice}
and~\ref{fig:defendchoice}, when $s_i$ are viewed as control states
and the bracketed parts of edge-labels are counter-changes (written in
binary).

First we consider 
the r-game $\calN^{csg}=(Q,Q_{\exists}, \gt{}, \{\WinE\})$ 
(``the control-state game of $\calN$'') arising from $\calN$ by
\emph{forgetting the counter-changes}; hence $q\gt{}\bar{q}$ iff 
there is a rule $q\gt{z}\bar{q}$. In fact, we will assume that 
there is at most one rule  $q\gt{z}\bar{q}$ in $\delta$ (of $\calN$) 
for any pair $(q,\bar{q})\in Q\times Q$; this can be achieved by
harmless modifications. 

We construct the LTS $\calL(\calN^{csg})$ 
(as defined by~(\ref{eq:ltstogame}) for general
r-games). Hence each $q\in Q$ has the copies $q,q'$ in
$\calL(\calN^{csg})$, and other states are added  
(as also depicted in Fig.~\ref{fig:defendchoice} where $s_i$ are
now in the role of control states); there are also the
respective labelled transitions in $\calL(\calN^{csg})$, with labels
$a_{\langle q,\bar{q}\rangle}$, $a_c$, $a_{win}$.

It remains to add the counter changes (integer increments and
decrements in binary), to create the required SOCN $\calN'$.
For $q\in Q_{\exists}$ this adding is simple, as
depicted in Fig.~\ref{fig:attchoice}: if $q\gt{z}\bar{q}$ (in $\calN$),
then we
simply extend the label $a_{\langle q,\bar{q}\rangle}$  
in $\calL(\calN^{csg})$
with $z$;
for $q\gtl{a_{\langle q,\bar{q}\rangle}}\bar{q}$ and
$q'\gtl{a_{\langle q,\bar{q}\rangle}}\bar{q}'$ in $\calL(\calN^{csg})$
we get 
 $q\gtl{a_{\langle q,\bar{q}\rangle},z}\bar{q}$
 and  $q'\gtl{a_{\langle q,\bar{q}\rangle},z}\bar{q}'$ in $\calN'$. 

 For $q\in Q_{\forall}$ (where $Q_{\forall}=Q\smallsetminus
 Q_{\exists}$) it is tempting to the same, i.e. extend
the label $a_{\langle q,\bar{q}\rangle}$ with $z$ when
$q\gt{z}\bar{q}$, and extend $a_c$ with $0$.
 But this might allow cheating for Defender: she could thus
 mimic choosing a
 transition
 $q(k)\gt{x}\bar{q}(k{+}x)$ even if $k{+}x<0$. This is avoided by the
 modification that is demonstrated in Fig.~\ref{fig:defendchoice}
 (by $x=x'{+}x''$, etc.); put simply: Defender must immediately prove
 that the transition she is choosing to mimic is indeed performable.
 Formally, if   $X=\{\bar{q}\mid q\gt{}\bar{q}\}\neq\emptyset$ 
 (in $\calL(\calN^{csg})$), then in $\calN'$ 
 we put $q\gtl{a_c,0}\langle q,X \rangle$ and
 $\langle q,X \rangle\gtl{a_{\langle q,\bar{q} \rangle},z}\bar{q}$
 for each $q\gt{z}\bar{q}$ (in $\calN$);
 for  each $q\gt{z}\bar{q}$ we also define $z'=\min\{z,0\}$,
  $z''=\max\{z,0\}$ and
 put $q'\gtl{a_c,z'}\langle q,\bar{q}\rangle$, 
 $\langle q,\bar{q}\rangle\gtl{a_{\langle
 q,\bar{q}\rangle},z''}\bar{q}'$.
 Then for any pair $q\gtl{\bar{z}}\bar{q}$, 
 $q\gtl{\bar{\bar{z}}}\bar{\bar{q}}$
 where $\bar{q}\neq \bar{\bar{q}}$ we put 
$\langle q,\bar{q}\rangle\gtl{a_{\langle
q,\bar{\bar{q}}\rangle},\bar{\bar{z}}-\bar{z}'}\bar{\bar{q}}$.

Finally, $p_{win}\gtl{a_{win}}p_{win}$ in $\calL(\calN^{csg})$
is extended to $p_{win}\gtl{a_{win},0}p_{win}$ in $\calN'$.
\end{proof}

\section{Additional remarks}~\label{sec:addrem}.

Our EXPSPACE-hardness proof of reachability games in Section~\ref{sec:reachgames}
can be easily adjusted to yield an alternative proof of
EXPTIME-hardness of countdown games
from~\cite{DBLP:journals/lmcs/JurdzinskiSL08}. 
The proof in~\cite{DBLP:journals/lmcs/JurdzinskiSL08} used a reduction 
from the
acceptance problem for alternating
linear bounded automata (a well-known EXPTIME-complete problem); a
crucial point was that a whole (linear) configuration can be
presented by an exponential number (presented in polynomial space),
and moving to a next configuration can be realized by adding (or,
 in another setting, subtracting) another exponential number.
The countdown 
games were modified to ``count-up'' games in~\cite{DBLP:journals/ipl/Kiefer13}  
(called ``hit-or-run games'' there) to show EXPTIME-hardness of
bisimilarity on BPA processes. 

In the case of countdown (or count-up) games an important ingredient
is an initial (or target) exponential counter value, as a part of the
input. In our proof such a value would correspond to an upper bound on
$t$ in the computation $C_0,C_1,\dots,C_t$ on $w$. We can imagine 
a modification of our game where 
 Eve first sets the corresponding bound-value into the counter and then
repeatedly decrements the counter before entering a configuration
$s_{(q_+,x)}(k)$. If we restrict our attention to 
Turing machines $M$ working in exponential time,
then constructing the modified game constitutes a (logspace, master) reduction
demonstrating EXPTIME-hardness.

If we only assume that $M$ works in exponential space, it can work 
in 
double-exponential time, and we cannot present such a bound  in
polynomial space (explicitly); therefore our proof does not rely on any
explicit upper bound, and it yields a modification of countdown games
where one player first sets any initial counter value and only then the
original countdown game is played. This modified game is thus
EXPSPACE-complete. (EXPSPACE-membership follows from our reduction in
Section~\ref{sec:reachtosimul} and the known results for simulation
and bisimulation relations, but also a simpler direct proof can be
given.) 

Finally we can note that our modified (EXPSPACE-complete) game does
not seem easily implementable by BPA processes, hence 
the EXPTIME-hardness result
in~\cite{DBLP:journals/ipl/Kiefer13} has not been improved here.
(The known upper bound for bisimilarity on BPA is 2-EXPTIME.)

\bibliographystyle{abbrv}
\bibliography{bibliography} 

\end{document}